\documentclass[11pt,a4paper]{article}
\usepackage{stmaryrd}
\usepackage{amsmath,amssymb,amsthm}
\usepackage{cite}
\usepackage{color}
\usepackage{hyperref}

\newtheorem{thm}{Theorem}[section]
\newtheorem{cor}[thm]{Corollary}
\newtheorem{lem}[thm]{Lemma}
\newtheorem{prop}[thm]{Proposition}
\theoremstyle{definition}

\theoremstyle{remark}

\numberwithin{equation}{section}

\newcommand{\bse}{\begin{subequations}}
\newcommand{\ese}{\end{subequations}}

\newcommand{\llb}{\llbracket}
\newcommand{\rlb}{\rrbracket}
\newcommand{\lpt}{\llparenthesis}
\newcommand{\rpt}{\rrparenthesis}
\newcommand{\lpb}{\{}
\newcommand{\rpb}{\}_\theta}

\def \t#1{\widetilde{#1}}

\def \c#1{\mathcal{#1}}

\title{Integrability properties of the dispersionless \\  Kadomtsev-Petviashvili hierarchy}

\author{
Wei Fu$^{1}$\footnote{E-mail address: wfu@shu.edu.cn},~~
R. Ilangovane$^{2}$,~~
K.M. Tamizhmani$^{2}$,~~
Da-jun Zhang$^{1}$\footnote{Corresponding author. E-mail address: djzhang@staff.shu.edu.cn}\\
{\small \it $^{1}$Department of Mathematics, Shanghai University, Shanghai 200444, P.R.China}\\
{\small \it $^{2}$Department of Mathematics, Pondicherry University, Puducherry  605014, India}
}

\date{\today}

\begin{document}

\maketitle

\begin{abstract}
In the paper we investigate integrability characteristics for the dispersionless Kadomtsev-Petviashvili hierarchy.
These characteristics include symmetries, Hamiltonian structures and conserved quantities.
We give a Lax triad to construct a master symmetry and a hierarchy of non-isospectral dispersionless Kadomtsev-Petviashvili flows.
These non-isospectral flows, together with the known isospectral dispersionless Kadomtsev-Petviashvili flows,
form a Lie algebra, which is used to derive two sets of symmetries for the isospectral dispersionless Kadomtsev-Petviashvili hierarchy.
By means of the master symmetry, symmetries, Noether operator and  conserved covariants, Hamiltonian structures
are constructed for both isospectral and non-isospectral dispersionless Kadomtsev-Petviashvili hierarchies.
Finally, two sets of conserved quantities and their Lie algebra are derived for the isospectral dispersionless Kadomtsev-Petviashvili hierarchy.
\vskip 6pt
\noindent{\textbf{Keywords:}} dispersionless Kadomtsev-Petviashvili hierarchy, symmetries, Hamiltonian structures, conserved quantities, Lie algebras.\\
\noindent{\textbf{PACS:}} 02.30.Ik
\end{abstract}

\section{Introduction}\label{sec:intro}

Compared with the  Kadomtsev-Petviashvili (KP) equation
\begin{align}\label{eq:KP}
u_{t}=\frac{1}{4}u_{xxx}+3uu_x+\frac{3}{4}\partial^{-1}_x u_{yy},
\end{align}
the dispersionless Kadomtsev-Petviashvili (dKP) equation
\begin{align}\label{eq:dKP}
U_{T}=3UU_X+\frac{3}{4}\partial^{-1}_XU_{YY}
\end{align}
is lack of the dispersion term $U_{XXX}$.
This kind of systems were first introduced in the study of Benney equations\cite{LM,Zak}.
The dKP hierarchy was derived by Kodama and Gibbons \cite{KG} via the Lax representations of the equations with a special
Poisson bracket in the consideration of the results in \cite{KM1,KM2}. They also produced a class of exact solutions including rarefaction 
waves and shock waves for the dKP equation \cite{Kod} as well as constructed  an infinite number of exact solutions
for the whole hierarchy using hodograph transform \cite{KG}.
At the same time, it was found by Krichever that the dKP hierarchy could be considered as a particular case of the general Whitham hierarchy \cite{Kri1}.
Then the  dKP hierarchy was applied to study the perturbed chiral primary rings of the topological minimal models \cite{Kri2}.
Motivated by Krichever's treatment in some sense, Takasaki and Takebe
reviewed the dKP hierarchy from the view point of Sato' approach and dispersionless limit\cite{TT1,TT2}.
Then, a series of literatures reported the relations of  dispersionless equations and topological field theory and so forth\cite{Dub,Car,Kri3,AK,TT3}.
Besides, the dKP equation or its hierarchy were also investigated from the aspects of
the Miura map \cite{CT}, reductions\cite{MAM1,MAM2}, inverse scattering transform\cite{MS} and so on.

The present paper aims to investigate integrability properties for the dKP hierarchy,
such as infinitely many symmetries, Hamiltonian structures and conserved quantities.
Let us recall the case of KP hierarchy.
Unlike the (1+1)-dimensional integrable systems which possess recursion operators with implectic-symplectic structure\cite{FF2,ZC1},
for the KP hierarchy, an easier way is to make use of its master symmetry\cite{Fuch}.
In this approach, the master symmetry was used as a flows generater to build
implicit recursion relations for the KP hierarchy,
which were then used to construct Hamiltonians, symmetries and conserved quantities\cite{OF,CLL,Cas1,Cas2}.
This approach was re-described systematically for a semi-discrete KP hierarchy in \cite{FHTZ} very recently.

In fact, master symmetries can be derived as integrable non-isospectral flows.
Unlike the traditional treatment for the single isospectral KP hierarchy
where one always takes $x\equiv t_1$ and $y\equiv t_2$ (see \cite{OSTT,CXZ}),
if we consider isospectral and non-isospectral cases together,
$x$ and $y$ have to be completely independent of $t_j$.
Therefore Lax triads are necessary in deriving the isospectral and non-isospectral KP hierarchies
because one has to consider $x$, $y$ and $t_j$ separately.
A detailed description for Lax triad approach can be found in \cite{FHTZ}.

The isospectral dKP hierarchy can be derived \cite{KG} from the following pseudo-polynomial of $P$,
\begin{align}\label{cL}
\c L=P+\sum_{j=1}^{\infty}U_{j+1}P^{-j}.
\end{align}
It is also known that the whole isospectral dKP hierarchy are related to the usual isospectral KP hierarchy
via certain limit procedure\cite{CT}.
In the present paper, to construct a set of non-isospectral dKP flows,
we will first examine the limit procedure for the Lax triad of the isospectral KP hierarchy.
Then, by imposing the same limit procedure on the non-isospectral KP case,
we find a triad for constructing the non-isospectral dKP flows,
and one of the non-isospectral dKP flows will act as the master symmetry.
Then the master symmetry is used to investigate symmetries,
Hamiltonian structures and conservation laws for the whole isospectral dKP hierarchy.
The non-isospectral dKP hierarchy are also shown to have Hamiltonian structures.

This paper is organized as follows. In Sec. \ref{sec:prelim} we introduce some basic notions
and Lax triads for deriving the isospectral and non-isospectral KP hierarchies.
In Sec. \ref{sec:hierar}, we derive both the isospectral and non-isospectral dKP hierarchies via the Lax representations.
 Sec. \ref{sec:alg} investigates algebraic relations of flows
 and constructs two sets of symmetries for the isospectral dKP hierarchy.
Then in Sec. \ref{sec:Hamilt} we investigate Hamiltonian structures and conserved quantities.

\section{Preliminary}\label{sec:prelim}

\subsection{Basic notions}\label{subsec:notat}

In this section we will briefly recall some notions of integrability characteristics.
Some notions considered in this section may depend in $C^{\infty}$-way in time parameter $t$.
For time-independent case, we only need to ignore the partial derivative term of those integrability characteristics with respect to time parameter $t$
and they will reduce to the notions introduced in \cite{FF2,Mag}.

For a given evolution equation
\begin{align}\label{def:eq}
u_t=K(u),~~u\in\c M,
\end{align}
where $\c M$ is an infinite dimensional linear manifold of $C^{\infty}$ functions $u(x,y)$ defined on $\mathbb{R}^2$
and vanishing rapidly at infinity, and $K$ is a vector field on $\c M$.
Here we require the solution of \eqref{def:eq} depending in $C^{\infty}$-way on time parameter $t$.
Because $\c M$ is linear, all fibers of the tangent bundle $\mathrm{T}{\c M}$ are copies of
the same vector space $\c S$, i.e., we can treat $\c M=\c S$.
However, it is convenient to regard them as different objects for a better geometrical understanding.
By $\c S^*$ we denote the dual space of $\c S$ w.r.t. the dual relation
\begin{align}\label{def:ip}
\langle f,g\rangle=\int_{-\infty}^{+\infty}\int_{-\infty}^{+\infty}f(x,y)g(x,y)\mathrm{d}x\mathrm{d}y,~~\forall f\in\c S^*,g\in\c S.
\end{align}
Besides, for an operator $T$ living on $S$ or $\c S$,
by $T^*$ we denote the adjoint operator of $T$ with respect to the bilinear form $\langle\cdot,\cdot\rangle$.

The standard commutator $\llb\cdot,\cdot\rlb$ of $C^{\infty}$ vector fields on $\c M$ is defined as
\begin{align}\label{def:Lie}
\llb f,g\rlb(u)=f'(u)[g(u)]-g'(u)[f(u)],~~f,g\in\c S,
\end{align}
where
\begin{align}\label{def:Gateaux}
f'(u)[g(u)]=\frac{\partial}{\partial\varepsilon}f(u+{\varepsilon}g(u))\Big|_{\varepsilon=0}
\end{align}
is the G\^ateaux derivative of $f$ in direction $g$ w.r.t. $u$,
and $f'$ is called the linearization operator of $f$.
If no confusion arises, we use $\llb f,g\rlb$ and $f'[g]$ instead of the notations in \eqref{def:Lie} and \eqref{def:Gateaux} respectively.

For a scalar field $H:\c M\times \mathbb{R} \to \mathbb{R},\, (u,t)\mapsto H(u,t)$
and a covector field $\gamma:\c M\times \mathbb{R} \to \c S^*,\,(u,t)\mapsto\gamma(u,t)$, if
\begin{align}\label{def:grad}
H'[g]=\langle\gamma,g\rangle,~~\forall g\in\c S,
\end{align}
then $\gamma$ is called the functional derivative or gradient of $H$, and $H$ is called the potential of $\gamma$.
Such a $\gamma$ is usually denoted by $\frac{\delta H}{\delta u}$ or $\mathrm{grad}\,H$.
\begin{prop}\label{prop:potent}\cite{FF2}
$\gamma=\gamma(u,t)\in\c S^*$ is a gradient field if and only if $\gamma'$ is a self-adjoint operator in terms of the dual relation \eqref{def:ip}, 
i.e., $\gamma'^*=\gamma'$. The corresponding  potential $H=H(u,t)$ can be given by
\begin{align}\label{def:potent}
H(u,t)=\int_{0}^{1}\langle\gamma(\lambda u,t),u\rangle\mathrm{d}\lambda.
\end{align}
\end{prop}

A vector field $G:\c M\times \mathbb{R}\to\c S,\,(u,t)\mapsto G(u,t)$ is a symmetry of \eqref{def:eq} if  
\begin{align}\label{def:sym}
G_t+\llb G,K\rlb=0
\end{align}
holds everywhere in $M\times \mathbb{R}$.
A covector field $\gamma:\c M\times\mathbb{R}\to\c S^*,\,(u,t)\mapsto\gamma(u,t)$ is called a conserved covariant of equation \eqref{def:eq} 
if 
\begin{align}\label{def:cc}
\gamma_t+\gamma'[K]+K'^*[\gamma]=0
\end{align}
holds everywhere in $M\times \mathbb{R}$.
A scalar field $H:\c M\times\mathbb{R}\to\mathbb{R},\,(u,t)\mapsto H(u,t)$ is called a conserved quantity of equation \eqref{def:eq} if 
\begin{align}\label{def:cq}
H_t+\Big\langle\frac{\delta H}{\delta u},K\Big\rangle=0
\end{align}
holds everywhere in $M\times \mathbb{R}$.
Conserved quantities and conserved covariants are closely related to each other (cf.\cite{FF2,FHTZ}).

\begin{prop}\label{prop:cq1}
Suppose that covector field $\gamma=\gamma(u,t)$ is a gradient field and scalar field $H=H(u,t)$ is its potential. 
Then, $H$ is a conserved quantity of equation \eqref{def:eq} if and only if $\gamma$ is a conserved covariant of \eqref{def:eq}.
\end{prop}


A linear operator $\theta(u):\c S^*\to\c S$ is called a Noether operator of equation \eqref{def:eq}, if
\begin{align}\label{def:Noeth}
\theta'[K]-\theta K'^*-K'\theta=0.
\end{align}
Noether operator $\theta$ maps conserved covariants of \eqref{def:eq} to its symmetries.

A linear operator $\theta(u):\c S^*\to\c S$ is called an implectic operator\cite{FF2} if it is skew-symmetric, i.e.,
$\theta(u)$ satisfying
\begin{align}\label{def:skew-sym}
\langle f,\theta g\rangle=-\langle\theta f,g\rangle 
\end{align}
and
\begin{align}\label{def:Jacobi}
\langle f,\theta'[\theta g]h\rangle+\langle g,\theta'[\theta h]f\rangle+\langle h,\theta'[\theta f]g\rangle=0,~~\forall f,g,h\in\c S^*.
\end{align}

The evolution equation \eqref{def:eq} has a Hamiltonian structure if it can be written in the form
\begin{align}\label{def:Hamilt}
u_t=K(u)=\theta(u)\frac{\delta H(u)}{\delta u},
\end{align}
where $\theta(u)$ is an implectic operator and $H(u)$ is called an Hamiltonian of equation \eqref{def:eq}.

\subsection{Lax triad and the KP hierarchy}\label{subsec:triad}

It is well known that the KP equation is connected with the following
pseudo-differential operator
\begin{align}\label{KP:L}
L=\partial+\sum_{j=1}^{\infty}u_{j+1}\partial^{-j},
\end{align}
where $\partial\doteq\partial_x$, $\partial \partial^{-1}=\partial^{-1}\partial=1$ and $u_j=u_j(x,y,\mathbf{t})\in\c M$
with time parameters $\mathbf{t}=(t_1,t_2,\cdots)$.
Traditionally, to get the KP equation one need to take $t_1\equiv x$ and $t_2\equiv y$ (cf. \cite{OSTT,CXZ}).
However, if one wants to derive the master symmetry as an integrable flow, then $y$ must be treated
as a new variable that is completely independent of $t_2$.
Thus, when we consider the KP hierarchy and the master symmetry simultaneously, Lax triad approach is necessary.
In \cite{FHTZ} we have shown that Lax triads play important roles in investigating integrability characteristics
for the KP and semi-discrete KP hierarchies.

Note that the dKP hierarchy can either be derived from the KP hierarchy through certain limit procedure,
or independently be derived from the Lax equations related to the operator $\mathcal{L}$.
In the following we will recall Lax triad approach for the isospectral and non-isospectral KP hierarchies.
The Lax representations of the non-isospectral KP hierarchy will be used to
lead to a non-isospectral dKP hierarchy under some limit procedure.

For the isospectral KP hierarchy we need
\bse
\label{iKP:triad}
\begin{align}
&L\phi=\eta\phi,~~\eta_{t_m}=0,\label{iKP:triad1}\\
&\phi_{y}=A_2\phi,~~A_2=\partial^2+2u_2,\label{iKP:triad2}\\
&\phi_{t_m}=A_m\phi,~~m=1,2,\cdots,\label{iKP:triad3}
\end{align}
\ese
where we suppose
\begin{align}
A_m=\partial^m+\sum_{j=1}^{m}a_j\partial^{m-j},~~A_m|_{\mathbf{u}=\mathbf{0}}=\partial^m,\label{iKP:A}
\end{align}
and $\mathbf{u}=(u_2,u_3,\cdots)$.
The coefficients $\{a_j\}$ can temporarily be left unknown.
The compatibility of \eqref{iKP:triad} reads
\bse\label{iKP:compat}
\begin{align}
&L_y=[A_2,L],\label{iKP:compat1}\\
&L_{t_m}=[A_m,L],\label{iKP:compat2}\\
&A_{2,t_m}-A_{m,y}+[A_2,A_m]=0,~~m=1,2,\cdots,\label{iKP:compat3}
\end{align}
\ese
where $[\cdot,\cdot]$ is defined as $[M,N]=MN-NM$.
Among the above compatibility conditions, \eqref{iKP:compat1} gives expressions of $\{u_j\}_{j>2}$ in terms of $u_2$, as the following,
\bse\label{KP:subst}
\begin{align}
&u_3=\frac{1}{2}(\partial^{-1}u_{2,y}-u_{2,x}),\label{KP:subst1}\\
&u_4=\frac{1}{4}(\partial^{-2}u_{2,yy}-2u_{2,y}+u_{2,xx}-2u_2^2),\label{KP:subst2}\\
&\cdots,\cdots.\nonumber
\end{align}
\ese
The equation \eqref{iKP:compat2} plays the role to determine those unknowns $\{a_j\}$ of $A_m$. In fact, $\{a_j\}$ can be 
uniquely determined from \eqref{iKP:compat2} and it turns out that $A_m$ is nothing but \cite{ZC2} $A_m=(L^m)_{+}$.
Here $(L^m)_{+}$ contains all the terms of $\partial^j$ with $j\geq 0$ in $L^m$. The first few of $A_m$ are
\bse\label{iKP:Am}
\begin{align}
&A_1=\partial,\label{iKP:A1}\\
&A_2=\partial^2+2u_2,\label{iKP:A2}\\
&A_3=\partial^3+3u_2\partial+3u_3+3u_{2,x},\label{iKP:A3}\\
&A_4=\partial^4+4u_2\partial^2+(4u_3+6u_{2,x})\partial+4u_4+6u_{3,x}+4u_{2,xx}+6u_2^2.\label{iKP:A4}
\end{align}
\ese
The third equation \eqref{iKP:compat3} provides the isospectral KP hierarchy (after replacing $\{u_j\}_{j\geq 3}$ by $u_2$ through \eqref{KP:subst})
\begin{align}\label{iKP:hierar}
u_{t_m}=K_m(u)=\frac{1}{2}({A}_{m,y}-[A_2,{A}_m]),~~m=1,2,\cdots,
\end{align}
where we also take $u_2=u$. $\{K_m\}$ are called isospectral KP flows.
Let us write down the first four equations in the KP hierarchy:
\bse\label{iKP:flow}
\begin{align}
&u_{t_1}=K_1(u)=u_x,\label{iKP:flow1}\\
&u_{t_2}=K_2(u)=u_y,\label{iKP:flow2}\\
&u_{t_3}=K_3(u)=\frac{1}{4}u_{xxx}+3uu_x+\frac{3}{4}\partial^{-1}u_{yy},\label{iKP:flow3}\\
&u_{t_4}=K_4(u)=\frac{1}{2}u_{xxy}+4uu_y+2u_x\partial^{-1}u_y+\frac{1}{2}\partial^{-2}u_{yyy},\label{iKP:flow4}
\end{align}
in which the third one gives the KP equation.
\ese

To derive a master symmetry we turn to the non-isospectral case in which we set
\begin{align}\label{nKP:eta}
\eta_{t_m}=\epsilon \eta^{m-1},~~m=1,2,\cdots,
\end{align}
where $\epsilon$ is a parameter which will play a key role in leading to dispersionless equations through some limit procedure.
In this turn the Lax triad reads
\bse\label{nKP:triad}
\begin{align}
&L\phi=\eta\phi,\label{nKP:triad1}\\
&\phi_{y}=A_2\phi,\label{nKP:triad2}\\
&\phi_{t_m}=B_m\phi,~~m=1,2,\cdots,\label{nKP:triad3}
\end{align}
\ese
and the compatibility is
\bse\label{nKP:compat}
\begin{align}
&L_y=[A_2,L],\label{nKP:compat1}\\
&L_{t_m}=[B_m,L]+\epsilon L^{m-1},\label{nKP:compat2}\\
&A_{2,t_m}-B_{m,y}+[A_2,B_m]=0,~~m=1,2,\cdots,\label{nKP:compat3}
\end{align}
\ese
where we suppose  $B_m$ is an undetermined operator of the form
\begin{align}\label{nKP:B}
B_m=\sum_{j=0}^mb_j\partial^{m-j}.
\end{align}
Checking the asymptotic results \eqref{nKP:compat2}$_{\mathbf{u}=\mathbf{0}}$ and \eqref{nKP:compat3}$_{\mathbf{u}=\mathbf{0}}$
respectively, one finds they  together give
the necessary asymptotic condition for $B_m$\footnote{We note that one can also add isospectral asymptotic terms, 
for example, $B_m|_{\mathbf{u}=\mathbf{0}}=2y\partial^m+x\partial^{m-1}+\partial^{m-2}$, when $m\geq 3$. 
This will lead to a non-isospectral flow combined by a isospectral flow $K_{m-2}$ and 
this does not change the basic algebraic structure of the flows (see \cite{FHTZ}).}:
\begin{align}\label{nKP:Bbc}
B_m|_{\mathbf{u}=\mathbf{0}}=\epsilon(2y\partial^m+x\partial^{m-1}),~~m=1,2,\cdots.
\end{align}
With the asymptotic condition \eqref{nKP:Bbc} the operator $B_m$ can uniquely be determined from \eqref{nKP:compat2} and the first few of them are
\bse\label{nKP:Bm}
\begin{align}
B_1&=\epsilon(2yA_1+x),\label{nKP:B1}\\
B_2&=\epsilon(2yA_2+xA_1),\label{nKP:B2}\\
B_3&=\epsilon(2yA_3+xA_2+(\partial^{-1}u_2)),\label{nKP:B3}\\
B_4&=\epsilon(2yA_4+xA_3+(\partial^{-1}u_2)\partial+2(\partial^{-1}u_3)),\label{nKP:B4}
\end{align}
\ese
where $\{A_j\}$ are given in \eqref{iKP:Am}.
Here we stress that the equation $\phi_y=A_2\phi$ appears in both the isospectral case \eqref{iKP:triad} and
the non-isospectral case \eqref{nKP:triad}.
This is because both cases share the same replacement relations \eqref{KP:subst} that results from $\phi_y=A_2\phi$,
which indicates the necessity of Lax triads.
Then, from \eqref{nKP:compat3} and using \eqref{KP:subst} we have the non-isospectral KP hierarchy
\begin{align}\label{nKP:hierar}
u_{t_m}=\sigma_{m}(u)=\frac{1}{2}(B_{m,y}-[A_2,B_m]),~~m=1,2,\cdots,
\end{align}
and the first four equations are
\bse\label{nKP:flow}
\begin{align}
u_{t_1}&=\sigma_1(u)=\epsilon\,2yK_1(u),\label{nKP:flow1}\\
u_{t_2}&=\sigma_2(u)=\epsilon(2yK_2(u)+xK_1(u)+2u),\label{nKP:flow2}\\
u_{t_3}&=\sigma_3(u)=\epsilon(2yK_3(u)+xK_2(u)+2\partial^{-1}u_{y}-u_x),\label{nKP:flow3}\\
u_{t_4}&=\sigma_4(u)=\epsilon\Big(2yK_4(u)+xK_3(u)+u_{xx}+4u^{2}+u_{x}\partial^{-1}u+\frac{3}{2}\partial^{-2}u_{yy}-\frac{3}{2}u_y\Big),\label{nKP:flow4}
\end{align}
\ese
where $\{K_j(u)\}$ are the isospectral flows given in \eqref{iKP:flow},
and we have taken $u_2=u$.
$\{\sigma_m(u)\}$ are the non-isospectral KP flows in which $\sigma_3$ (with $\epsilon=1$) is the master symmetry of the KP equation\cite{OF}.

\section{The dKP hierarchies}\label{sec:hierar}

\subsection{The isospectral dKP hierarchy}\label{subsec:idKP}

It is known that the dKP hierarchy is related to the KP hierarchy through certain limit procedure and
the result copes with the derivation from the operator $\mathcal{L}$  \cite{CT}.
In the following let us repeat the same derivation and limit procedure for the isospectral dKP hierarchy
but here we will start from a triad.
This will help us to construct a master symmetry and non-isospectral flows for the dispersionless case.

\subsubsection{Derivation from $\mathcal{L}$}

We suppose $U=U(X,Y,\mathbf{T})\in\c M$ with time parameters $\mathbf{T}=(T_1,T_2,T_3,\cdots)$
and $\c L$ is a pseudo-polynomial of $P$ defined as \eqref{cL}, i.e.,
\begin{align}\label{dKP:L}
\c L=P+\sum_{j=1}^{\infty}U_{j+1}P^{-j},
\end{align}
where $U_{j}=U_j(X,Y,\mathbf{T})\in\c M$, parameters $\mathbf{T}=(T_1,T_2,T_3,\cdots)$
and $P$ is independent of $(X, Y,\mathbf{T})$,
and $\c A_m$ is a polynomial of the form
\begin{align}\label{idKP:A}
\c A_m=P^m+\sum_{j=1}^m\t a_jP^{m-j}
\end{align}
with asymptotic condition
\begin{align}\label{idKP:Abc}
\c A_m|_{\mathbf{U}=\mathbf{0}}=P^m,
\end{align}
where $\mathbf{U}=(U_2,U_3,\cdots)$. Now one can consider the triad
\bse\label{idKP:compat}
\begin{align}
&\c L_Y=\{\c A_2,\c L\},~~\c A_2=P^2+2U_2,\label{idKP:compat1}\\
&\c L_{T_m}=\{\c A_m,\c L\},\label{idKP:compat2}\\
&\c A_{2,T_m}-\c A_{m,Y}+\{\c A_2,\c A_m\}=0,~~m=1,2,\cdots,\label{idKP:compat3}
\end{align}
\ese
where the Poisson bracket $\{\cdot,\cdot\}$ is defined by\cite{KG}
\begin{align}\label{def:Poiss1}
\{F,G\}=\frac{\partial F}{\partial P}\frac{\partial G}{\partial X}-\frac{\partial F}{\partial X}\frac{\partial G}{\partial P}.
\end{align}
Under the condition \eqref{idKP:Abc}, $\c A_m$ can then be uniquely determined from  \eqref{idKP:compat2} by comparing the positive powers of $P$,
and it turns out that $\c A_m=(\c L^m)_+$ in terms of $P$. The first few of the polynomial $A_m$ are
\bse\label{idKP:Am}
\begin{align}
\c A_1&=P,\label{idKP:A1}\\
\c A_2&=P^2+2U_2,\label{idKP:A2}\\
\c A_3&=P^3+3U_2P+3U_3,\label{idKP:A3}\\
\c A_4&=P^4+4U_2P^2+4U_3P+4U_4+6U_2^2.\label{idKP:A4}
\end{align}
\ese
Still by comparing the powers of $P$, the equation \eqref{idKP:compat1} in the triad \eqref{idKP:compat}
 contributes the replacement relations between $\{U_j\}_{j\geq3}$ and $U_2$, which are
\bse\label{dKP:subst}
\begin{align}
&U_3=\frac{1}{2}\partial_X^{-1}U_{2,Y},\label{dKP:subst1}\\
&U_4=\frac{1}{4}(\partial_X^{-2}U_{2,YY}-2U_2^2),\label{dKP:subst2}\\
&\cdots,\cdots.\nonumber
\end{align}
\ese
Finally, the dKP hierarchy is derived from  \eqref{idKP:compat3} and written as
\begin{align}\label{idKP:hierar1}
U_{T_m}=\t K_m(U)=\frac{1}{2}(\c A_{m,Y}-\{\c A_2,\c A_m\}).
\end{align}
Here we have replaced $\{U_j\}_{j\geq3}$ by $U_2$ under \eqref{dKP:subst}
and made use of $\c A_{2,T_m}=\c A_2'[U_{T_m}]=2U_{T_m}$.
$\{\t K_m\}$ are called isospectral dKP flows.
The first few equations in the isospectral dKP hierarchy are
\bse\label{idKP:flow}
\begin{align}
U_{T_1}&=\t K_1(U)=U_X,\label{idKP:flow1}\\
U_{T_2}&=\t K_2(U)=U_Y,\label{idKP:flow2}\\
U_{T_3}&=\t K_3(U)=3UU_X+\frac{3}{4}\partial_X^{-1}U_{YY},\label{idKP:flow3}\\
U_{T_4}&=\t K_4(U)=4UU_Y+2U_X\partial_X^{-1}U_Y+\frac{1}{2}\partial_X^{-2} U_{YYY},\label{idKP:flow4}
\end{align}
\ese
where the third one is the dKP equation.

\subsubsection{Limit procedure}

The dKP hierarchy \eqref{idKP:hierar1} is related to the KP hierarchy \eqref{iKP:hierar} through
some formal limit procedure \cite{CT}.

Let us introduce the  scalar transform relations
\begin{align}\label{lim:xyt}
X=\epsilon x,~~Y=\epsilon y,~~\mathbf{T}=(T_1,T_2,\cdots)=\epsilon\mathbf{t}=(\epsilon t_1,\epsilon t_2,\cdots),
\end{align}
under which we write
\begin{align*}
u_j(x,y,\mathbf{t})=U_j(X,Y,\mathbf{T}),~~j=2,3,\cdots.
\end{align*}
For the wave function $\phi$ in the Wentzel-Kramers-Brillouin (WKB) form with the action $S$, i.e.,
\begin{align}\label{lim:phi}
\phi=\mathrm{exp}\bigg[\frac{1}{\epsilon}S(X,Y,\mathbf{T},\lambda)\bigg],
\end{align}
it can be seen that
\begin{align}\label{lim:part}
(\epsilon\partial_X)^j\phi=P^j\phi+o(\epsilon),~~\forall j\in\mathbb{Z},
\end{align}
where $P=S_X$. Besides, in terms of $P$ one can find that
\bse\label{lim:LA}
\begin{align}
&L\phi=\c L\phi+o(\epsilon),\label{lim:L}\\
&A_m\phi=\c A_m\phi+o(\epsilon),\label{lim:A}
\end{align}
\ese
where $\c L$ and $\c A_m$ are defined as before.
Then, for the triad \eqref{iKP:compat}, we have
\begin{align*}
(L_y-[A_2,L])\phi={}&\epsilon\bigg[\big(\sum_{j=1}^{\infty}U_{j+1,Y}P^{-j}\big)-2P\big(\sum_{j=1}^{\infty}U_{j+1,X}P^j\big) \\
&+2U_{2,X}\big(1-\sum_{j=1}^{\infty}jU_{j+1}P^{-j-1}\big)\bigg]\phi+o(\epsilon^2),
\end{align*}
\begin{align*}
(L_{t_m}-[A_m,L])\phi={}&\epsilon\bigg[\big(\sum_{j=1}^{\infty}U_{j+1,T_m}P^{-j}\big)\\
&-\big(mP^{m-1}+\sum_{j=1}^{m-1}(m-j)\t a_jP^{m-j-1}\big)\big(\sum_{j=1}^{\infty}U_{j+1,X}P^j\big) \\
&+\big(\sum_{j=1}^m\t a_{j,X}P^{m-j}\big)\big(1-\sum_{j=1}^{\infty}jU_{j+1}P^{-j-1}\big)\bigg]\phi+o(\epsilon^2),
\end{align*}
\begin{align*}
(A_{2,t_m}-A_{m,y}+[A_2,A_m])\phi={}&\epsilon\bigg[2U_{2,T_m}-\sum_{j=1}^m\t a_{j,Y}P^{m-j}+2P(\sum_{j=1}^m\t a_{j,X}P^{m-j})\\
&-2U_{2,X}\big(mP^{m-1}+\sum_{j=1}^{m-1}(m-j)\t a_jP^{m-j-1}\big)\bigg]\phi+o(\epsilon^2),
\end{align*}
for $m=1,2,\cdots$,
where, if we formally consider $P$ to be independent of $(X,Y,\mathbf{T})$ and make use of the Poisson bracket $\{\cdot,\cdot\}$ defined in \eqref{def:Poiss1},
the leading terms on the right hand side of each equation can be written as
\begin{align*}
& (\c L_Y-\{\c A_2,\c L\})\phi, \\
& (\c L_{T_m}-\{\c A_m,\c L\})\phi, \\
& (\c A_{2,T_m}-\c A_{m,Y}+\{\c A_2,\c A_m\})\phi,
\end{align*}
which should be zero.
This then copes with the triad \eqref{idKP:compat}.

\subsection{The non-isospectral dKP hierarchy and master symmetry}\label{subsec:ndKP}

Now the plan is clear. We impose the same limit procedure on the non-isospectral triad \eqref{nKP:compat}
(acting on $\phi$)
and also make use of the Poisson bracket $\{\cdot,\cdot\}$.
After checking the leading terms,
we can find a triad starting from $\c L$ to derive the non-isospectral dKP flows and an integrable master symmetry flow.
The results can be described as follows.

We can start from
\bse
\label{ndKP:compat}
\begin{align}
&\c L_Y=\{\c A_2,\c L\},~~\c A_2=P^2+2U_2,\label{ndKP:compat1}\\
&\c L_{T_m}=\{\c B_m,\c L\}+\c L^{m-1},\label{ndKP:compat2}\\
&\c A_{2,T_m}-\c B_{m,Y}+\{\c A_2,\c B_m\}=0,~~m=1,2,\cdots,\label{ndKP:compat3}
\end{align}
\ese
where $\c L$ is defined as before, and
\begin{align}\label{ndKP:B}
\c B_m=\sum_{j=0}^m\t b_jP^{m-j}
\end{align}
is an undetermined polynomial of $P$ living on $\c M$. Consider \eqref{ndKP:compat2} and \eqref{ndKP:compat3} asymptotically, 
i.e., \eqref{ndKP:compat2}$|_{\mathbf{U}=\mathbf{0}}$ and \eqref{ndKP:compat3}$|_{\mathbf{U}=\mathbf{0}}$, one can find
\begin{align*}
&(\t b_j|_{\mathbf{U}=\mathbf{0}})_X=
\left\{
\begin{array}{ll}
1, & j=1, \\
0, & j=0~\hbox{or}~j=2,3,\cdots,m,
\end{array}
\right.\\
&(\t b_j|_{\mathbf{U}=\mathbf{0}})_Y=
\left\{
\begin{array}{ll}
2, & j=0, \\
0, & j=1,2,\cdots,m,
\end{array}
\right.
\end{align*}
which determines the asymptotic condition for $\c B_m$ as\footnote{Like the non-isospectral KP hierarchy,
one can also add the isospectral asymptotic terms in $B_m$, e.g., $\c B_m|_{\mathbf{U}=\mathbf{0}}=2YP^m+XP^{m-1}+P^{m-2}$,
but \eqref{ndKP:Bbc} copes with the leading term  of \eqref{nKP:Bbc} in the limit procedure.}
\begin{align}\label{ndKP:Bbc}
\c B_m|_{\mathbf{U}=\mathbf{0}}=2YP^m+XP^{m-1},~~m=1,2,\cdots.
\end{align}
Then \eqref{ndKP:compat2} with \eqref{ndKP:Bbc} determine the expression of $\c B_m$ uniquely. Here we list out the first few of $\c B_m$,
\bse\label{ndKP:Bm}
\begin{align}
\c B_1&=2Y\c A_1+X,\label{ndKP:B1}\\
\c B_2&=2Y\c A_2+X\c A_1,\label{ndKP:B2}\\
\c B_3&=2Y\c A_3+X\c A_2+\partial_X^{-1}U_2,\label{ndKP:B3}\\
\c B_4&=2Y\c A_4+X\c A_3+(\partial_X^{-1}U_2)P+\partial_X^{-1}U_3,\label{ndKP:B4}
\end{align}
\ese
where $\{\c A_j\}$ are given in \eqref{idKP:Am}.
Equation \eqref{ndKP:compat1} gives the same replacement relations \eqref{dKP:subst} and with the help of it, 
\eqref{ndKP:compat3} provides the non-isospectral dKP hierarchy (with $U_2=U$)
\begin{align}\label{ndKP:hierar1}
U_{T_m}=\t\sigma_m(U)=\frac{1}{2}(\c B_{m,Y}-\{\c A_2,\c B_m\}).
\end{align}
The first few non-isospectral dKP equations are
\bse\label{ndKP:flow}
\begin{align}
U_{T_1}&=\t\sigma_1(U)=2Y\t K_1(U),\label{ndKP:flow1}\\
U_{T_2}&=\t\sigma_2(U)=2Y\t K_2(U)+X\t K_1(U)+2U,\label{ndKP:flow2}\\
U_{T_3}&=\t\sigma_3(U)=2Y\t K_3(U)+X\t K_2(U)+2\partial_X^{-1}U_Y,\label{ndKP:flow3}\\
U_{T_4}&=\t\sigma_4(U)=2Y\t K_4(U)+X\t K_3(U)+4U^2+U_X\partial_X^{-1}U+\frac{3}{2}\partial_X^{-2}U_{YY},\label{ndKP:flow4}
\end{align}
\ese
where $\{\t K_j(U)\}$ have been given in \eqref{idKP:flow}.
$\{\t\sigma_m(U)\}$ defined by \eqref{ndKP:hierar1} are called non-isospectral dKP flows,
and later we can see that $\t\sigma_3$ plays a role of master symmetry.

\subsection{Lax representations of dKP flows}\label{subsec:dKP}

Now we conclude the main results in this section as the following proposition.
\begin{prop}\label{prop:hierar1}
The isospectral dKP flows $\{\t K_s(U)\}$ and the non-isospectral dKP flows $\{\t \sigma_s(U)\}$
can be expressed through
\bse\label{dKP:hierar2}
\begin{align}
\t K_s(U)&=\frac{1}{2}(\c A_{s,Y}-\{\c A_2,\c A_s\}),\label{idKP:hierar2}\\
\t\sigma_s(U)&=\frac{1}{2}(\c B_{s,Y}-\{\c A_2,\c B_s\})\label{ndKP:hierar2},
\end{align}
\ese
respectively, with the asymptotic conditions
\bse\label{dKP:asym}
\begin{align}
&\t K_s(U)|_{U=0}=0,~~\c A_s|_{U=0}=0,\label{idKP:asym}\\
&\t\sigma_s(U)|_{U=0}=0,~~\c B_s|_{U=0}=0,\label{ndKP:asym}
\end{align}
\ese
for $s=1,2,\cdots$.
Corresponding to the case of the KP hierarchy, we call \eqref{dKP:hierar2}
the Lax representations of $\{\t K_s(U)\}$ and   $\{\t \sigma_s(U)\}$, respectively.
\end{prop}

Besides, the isospectral dKP flows $\{\t K_s(U)\}$ can be expressed in terms of $\c L$. 
\begin{prop}\label{prop:hierar2}
The isospectral dKP flows $\{\t K_s(U)\}$ defined in \eqref{idKP:hierar2} can be expressed as
\begin{align}\label{idKP:res}
\t K_s(U)=\partial_X\underset{P}{\mathrm{Res\,}}\c L^s,
\end{align}
where
\begin{align*}
\underset{P}{\mathrm{Res\,}}\bigg(\sum^{+\infty}_{j=-m}\t c_j P^j\bigg) =\t c_{-1},~~(m\geq1).
\end{align*}
\end{prop}

\begin{proof}
From \eqref{idKP:hierar2} we have
\begin{align*}
2\t K_s(U)&=\c A_{s,Y}-\{\c A_2,\c A_s\}\\
&=[(\c L^s-(\c L^s)_{-})_Y-\{\c A_2,\c L^s-(\c L^s)_{-}\}]_0\\
&=[(\c L^s)_Y-\{\c A_2,\c L^s\}-((\c L^s)_{-})_Y+\{\c A_2,(\c L^s)_{-}\}]_0.
\end{align*}
Here $(\c L^s)_{-}=\c L^s-(\c L^s)_{+}$, and $[\,\cdot\,]_0$  means taking the constant part of the polynomial  $[\,\cdot\,]$.
Noting that \eqref{idKP:compat1} indicates $(\c L^s)_Y-\{\c A_2,\c L^s\}=0$,
we then have
\begin{align*}
2\t K_s(U)=\{\c A_2,(\c L^s)_{-}\}_0=2\partial_X\underset{P}{\mathrm{Res\,}}\c L^s
\end{align*}
and the proof is finished.
\end{proof}

\section{Algebra of flows, recursion relations and symmetries}\label{sec:alg}

The dKP flows $\{\t K_l(U)\}$ and $\{\t\sigma_r(U)\}$ span a Lie algebra with respect to the commutator $\llb\cdot,\cdot\rlb$.
In order to prove that, let us start from  the following two lemmas.
\begin{lem}\label{lem:alg1}
Suppose that $\t X\in\c S$ and
\begin{align}\label{lem:alg11}
\c E=\t d_0P^m+\t d_1P^{m-1}+\cdots+\t d_{m-1}P+\t d_m
\end{align}
is a polynomial of $P$ living on $\c M$  with asymptotic condition (i.e., each $\t d_j\in\c M$)
\begin{align}\label{lem:alg12}
\c E|_{U=0}=0,
\end{align}
then the equation
\begin{align}\label{lem:alg13}
2\t X=\c E_{Y}-\{\c A_2,\c E\}
\end{align}
has only the zero solution $\t X=0,\,\c E=0$. Here $\c A_2=P^2+2U$ and we have taken $U_2=U$.
\end{lem}

\begin{proof}
Comparing the highest power of $P$ in \eqref{lem:alg13}, we immediately  find $\t d_{0,X}=0$.
In light of \eqref{lem:alg12} which implies $\t d_0|_{U=0}=0$, $\t d_0$ can only be zero.
Step by step, we have $\t d_1=\t d_2=\cdots=\t d_m=0$ which leads to $\c E=0$.
Consequently, we have $\t X=0$ from \eqref{lem:alg13}.
\end{proof}

\begin{lem}\label{lem:alg2}
The isospectral dKP flows $\{\t K_l(U)\}$ and the non-isospectral dKP flows $\{\t\sigma_r(U)\}$, and polynomials $\{\c A_l\}$ and $\{\c B_r\}$ satisfy
\bse\label{lem:alg21}
\begin{align}
&2\llb{\t K_{l},\t K_{r}}\rlb=\lpt\c A_{l},\c A_{r}\rpt_Y-\{\c A_2,\lpt\c A_{l},\c A_{r}\rpt\}\label{lem:alg21a},\\
&2\llb{\t K_{l},\t\sigma_{r}}\rlb=\lpt\c A_{l},\c B_{r}\rpt_Y-\{\c A_2,\lpt\c A_{l},\c B_{r}\rpt\},\label{lem:alg21b}\\
&2\llb{\t\sigma_{l},\t\sigma_{r}}\rlb=\lpt\c B_{l},\c B_{r}\rpt_Y-\{\c A_2,\lpt\c B_{l},\c B_{r}\rpt\}\label{lem:alg21c},
\end{align}
\ese
where
\bse\label{lem:alg22}
\begin{align}
&\lpt\c A_{l},\c A_{r}\rpt=\c A_{l}'[\t K_{r}]-\c A_{r}'[\t K_{l}]+\{\c A_{l},\c A_{r}\}\label{lem:alg22a},\\
&\lpt\c A_{l},\c B_{r}\rpt=\c A_{l}'[\t\sigma_{r}]-\c B_{r}'[\t K_{l}]+\{\c A_{l},\c B_{r}\},\label{lem:alg22b}\\
&\lpt\c B_{l},\c B_{r}\rpt=\c B_{l}'[\t\sigma_{r}]-\c B_{r}'[\t\sigma_{l}]+\{\c B_{l},\c B_{r}\},\label{lem:alg22c}
\end{align}
\ese
with asymptotic conditions ($U=U_2$)
\bse\label{lem:alg23}
\begin{align}
&\lpt\c A_{l},\c A_{r}\rpt|_{U=0}=0,\label{lem:alg23a}\\
&\lpt\c A_{l},\c B_{r}\rpt|_{U=0}=l\,P^{l+r-2},\label{lem:alg23b}\\
&\lpt\c B_{l},\c B_{r}\rpt|_{U=0}=(l-r)(2YP^{l+r-2}+XP^{l+r-3})\label{lem:alg23c}.
\end{align}
\ese
\end{lem}

\begin{proof}
We only give the proof of \eqref{lem:alg21b}. \eqref{lem:alg21a} and \eqref{lem:alg21c} can be proved in a similar way.
Using the Lax representations \eqref{idKP:hierar2} and \eqref{ndKP:hierar2}, we can have that
\begin{align*}
2\t K_l'[\t\sigma_r]&=(\c A_{l,Y}-\{\c A_2,\c A_l\})'[\t\sigma_r]\\
&=(\c A_l'[\t\sigma_r])_Y-\{\c A_2'[\t\sigma_r],\c A_l\}-\{\c A_2,\c A_l'[\t\sigma_r]\}\\
&=(\c A_l'[\t\sigma_r])_Y-\{\c B_{r,Y}-\{\c A_2,\c B_r\},\c A_l\}-\{\c A_2,\c A_l'[\t\sigma_r]\}\\
&=(\c A_l'[\t\sigma_r])_Y+\{\c A_l,\c B_{r,Y}\}-\{\c A_2,\c A_l'[\t\sigma_r]\}+\{\c A_l,\{\c B_r,\c A_2\}\}
\end{align*}
and
\begin{align*}
2\t\sigma_r'[\t K_l]&=(\c B_{r,Y}-\{\c A_2,\c B_r\})'[\t K_l]\\
&=(\c B_r'[\t K_l])_Y-\{\c A_2'[\t K_l],\c B_r\}-\{\c A_2,\c B_r'[\t K_l]\}\\
&=(\c B_r'[\t K_l])_Y-\{\c A_{l,Y}-\{\c A_2,\c A_l\},\c B_r\}-\{\c A_2,\c B_r'[\t K_l]\}\\
&=(\c B_r'[\t K_l])_Y-\{\c A_{l,Y},\c B_r\}-\{\c A_2,\c B_r'[\t K_l]\}-\{\c B_r,\{\c A_2,\c A_l\}\}.
\end{align*}
Then, by subtraction of the above two equations, we reach to \eqref{lem:alg21b},
 where the Jacobi identity
\begin{align*}
\{\c A_2,\{\c A_{l},\c B_{r}\}\}+\{\c A_l,\{\c B_r,\c A_2\}\}+\{\c B_r,\{\c A_2,\c A_l\}\}=0
\end{align*}
is used.
\eqref{lem:alg23b} can be checked  under the asymptotic condition \eqref{dKP:asym}.
We note that the method to prove this lemma has been used for many integrable systems (e.g., \cite{CZ,Ma,ZNBC}).
\end{proof}

Lemma \ref{lem:alg1} and Lemma \ref{lem:alg2} lead to the following main theorem.
\begin{thm}\label{thm:alg1}
The isospectral dKP flows $\{\t K_l(U)\}$ and the non-isospectral dKP flows $\{\t\sigma_r(U)\}$ span a Lie algebra with basic structure
\bse\label{dKP:alg1}
\begin{align}
&\llb\t K_l,\t K_r\rlb=0,\label{dKP:alg1a}\\
&\llb\t K_l,\t\sigma_r\rlb=l\,\t K_{l+r-2},\label{dKP:alg1b}\\
&\llb\t\sigma_l,\t\sigma_r\rlb=(l-r)\t\sigma_{l+r-2},\label{dKP:alg1c}
\end{align}
\ese
where $l,r\geq 1$ and we set $\t K_0(U)=\t\sigma_0(U)=0$.
\end{thm}

\begin{proof}
We only prove \eqref{dKP:alg1b} and the proofs for the rest equations are similar. Let
\begin{align*}
\t X=\llb\t K_l,\t\sigma_r\rlb-l\,\t K_{l+r-2},
\end{align*}
we have
\begin{align*}
2\t X=(\lpt\c A_{l},\c B_{r}\rpt-l\,\c A_{l+r-2})_Y-\{\c A_2,\lpt\c A_{l},\c B_{r}\rpt-l\,\c A_{l+r-2}\}
\end{align*}
with
\begin{align*}
\lpt\c A_{l},\c B_{r}\rpt-l\,\c A_{l+r-2}|_{U=0}=0
\end{align*}
by Lemma \ref{lem:alg2} and Proposition \ref{prop:hierar1}. Then, $\t X$ must be zero under the result of Lemma \ref{lem:alg1}, 
which implies that \eqref{dKP:alg1b} holds.
\end{proof}

With the help of Theorem \ref{thm:alg1}, we have
\begin{thm}\label{thm:alg2}
Each equation
\begin{align}\label{thm:alg21}
U_{T_s}=\t K_s(U)
\end{align}
in the isospectral dKP hierarchy \eqref{idKP:hierar1} has two sets of infinitely many symmetries
\begin{align}\label{idKP:sym}
\{\t K_{l}(U)\},~~\{\t\tau_r^s(U,T_s)=sT_s\t K_{s+r-2}(U)+\t\sigma_r(U)\}
\end{align}
and they span a Lie algebra with basic structure
\bse\label{dKP:alg2}
\begin{align}
&\llb\t K_l,\t K_r\rlb= 0,\label{dKP:alg2a}\\
&\llb\t K_l,\t\tau_r^s\rlb= l\,\t K_{l+r-2},\label{dKP:alg2b}\\
&\llb\t \tau_l^s,\t\tau_r^s\rlb=(l-r)\t\tau_{l+r-2}^s,\label{dKP:alg2c}
\end{align}
\ese
where $l,r,s\geq 1$ and we set $\t K_0(U)=\t\tau_0^s(U)=0$.
\end{thm}

\begin{proof}
It is clear that $\{\t K_l\}$ are symmetries of the equation \eqref{thm:alg21}
due to the algebraic relation  \eqref{dKP:alg1a} and the definition \eqref{def:sym}.
For $\{\t\tau_r^s\}$ one can verify that it obeys the definition \eqref{def:sym}
when the relations \eqref{dKP:alg1} holds. Finally, \eqref{dKP:alg2} can be easily derived from \eqref{dKP:alg1}.
\end{proof}

Finally we note that the result of Theorem \ref{thm:alg1} can be understood as a recursive relation of the dKP flows.
The non-isospectral flow $\t \sigma_3$ acts as a role of flows generator as well as a master symmetry.
\begin{prop}\label{prop:recur}
The master symmetry $\t\sigma_3$ acts as a flow generator via the following relations
\bse
\begin{align}
&\t K_s=\frac{1}{s}\llb K_{s-1},\t\sigma_3\rlb,~~s>1,\label{idKP:recur}\\
&\t\sigma_s=\frac{1}{s-4}\llb \t\sigma_{s-1},\t\sigma_3\rlb,~~s>1,s\neq4\label{ndKP:recur}
\end{align}
\ese
with initial flows $\t K_1(U)$ given in \eqref{idKP:flow} and $\t\sigma_1(U),\t\sigma_4(U)$ given in \eqref{ndKP:flow}.
\end{prop}

\section{Hamiltonian structures and conserved quantities}\label{sec:Hamilt}

It is obvious that the dKP equation \eqref{idKP:flow3} has a Hamiltonian structure with Hamiltonian operator $\partial_X$,
i.e., the dKP equation \eqref{idKP:flow3} can be written as
\begin{equation}
U_T=\partial_X\Bigl(\frac{3}{2}U^2+\frac{3}{4}\partial_X^{-2}U_{YY}\Bigr)=\partial_X\frac{\delta \c H(U)}{\delta U},
\end{equation}
where $\frac{\delta \c H(U)}{\delta U}=\frac{3}{2}U^2+\frac{3}{4}\partial^{-2}U_{YY}$ is a gradient.
Unlike the case in (1+1)-dimensional systems, in which a recursion operator
with symplectic-implectic structure plays an important role in its integrability analysis\cite{FF2},
here the recursion relation given by $\llb\cdot,\t\sigma_3\rlb$ will be the key to the Hamiltonian structures of the dKP hierarchy.

Let us start from two lemmas. First,
\begin{lem}\label{lem:Hamilt1}
The following formula
\begin{align}\label{dKP:gradeq}
\mathrm{grad}\langle\t\gamma,\t\sigma\rangle=\t\gamma'^*\t\sigma+\t\sigma'^*\t\gamma
\end{align}
holds for any $\t\gamma\in\c S^*,\t\sigma\in\c S$.
\end{lem}

\begin{proof}
The proof is direct.
 For any function $g$ in $\c S$, the equation
\begin{align*}
\langle\t\gamma,\t\sigma\rangle'[g]=\langle\t\gamma'[g],
\t\sigma\rangle+\langle\t\gamma,\t\sigma'[g]\rangle=\langle\t\gamma'^*\t\sigma+\t\sigma'^*\t\gamma,g\rangle
\end{align*}
holds, which gives \eqref{dKP:gradeq} in light of \eqref{def:grad}.
\end{proof}

For the second lemma, we need to check it by lengthy direct calculation.
We skip the proof and  the lemma is
\begin{lem}\label{lem:Hamilt2}
$\partial_X$ is a Noether operator of the master symmetry equation $U_{T_3}=\t\sigma_3(U)$,
i.e.,
\begin{align}\label{dKP:masteq}
\t\sigma_3'\partial_X+\partial_X\t\sigma_3'^*=0.
\end{align}
\end{lem}

With the above two Lemmas, we come to the main results of this section.
\begin{thm}\label{thm:Hamilt}
(1). Each equation in the isospectral dKP hierarchy \eqref{idKP:hierar1} has a Hamiltonian structure
\begin{align}\label{idKP:Hamilt}
U_{T_s}=\t K_s(U)=\partial_X\frac{\delta\c H_s(U)}{\delta U},
\end{align}
where the gradient field $\t\gamma_s=\frac{\delta\c H_s}{\delta U}$  is defined by
\begin{align}\label{idKP:cc}
\t\gamma_s(U)=
\left\{\begin{array}{ll}
U, & s=1,\\
\frac{1}{s-1}{\rm grad}\langle\t\gamma_{s-1}(U),\t\sigma_3(U)\rangle, & s>1,
\end{array}\right.
\end{align}
and the Hamiltonian is
\begin{align}\label{idKP:H}
\c H_s(U)=
\left\{\begin{array}{ll}
\frac{1}{2}\langle U,U\rangle, & s=1,\\
\frac{1}{s-1}\langle\t\gamma_{s-1}(U),\t\sigma_3(U)\rangle, & s>1.\\
\end{array}\right.
\end{align}
(2). Each equation in the non-isospectral dKP hierarchy \eqref{ndKP:hierar1} also has a Hamiltonian structure
\begin{align}\label{ndKP:Hamilt}
U_{T_s}=\t\sigma_s(U)=\partial_X\frac{\delta\c J_s(U)}{\delta U},
\end{align}
where the gradient field $\t\omega_s=\frac{\delta\c J_s}{\delta U}$  is defined by
\begin{align}\label{ndKP:cc}
\t\omega_s(U)=
\left\{\begin{array}{ll}
2YU, & s=1,\\
\frac{1}{s-1}{\rm grad}\langle\t\omega_{s-1}(U),\t\sigma_3(U)\rangle, & s>1,\\
2Y\t\gamma_4(U)+X\t\gamma_3(U)+\frac{3}{2}\partial_X^{-1}U^2+\frac{3}{4}\partial_X^{-3}U_{YY}+U\partial_X^{-1}U, & s=4,
\end{array}\right.
\end{align}
and the Hamiltonian is
\begin{align}\label{ndKP:J}
\c J_s(U)=
\left\{\begin{array}{ll}
\langle YU,U\rangle, & s=1,\\
\frac{1}{s-4}\langle\t\omega_{s-1}(U),\t\sigma_3(U)\rangle, & s>1, s\neq 4,\\
\int_0^1\langle\t\omega_4(\lambda U),U\rangle\mathrm{d}\lambda, & s=4.
\end{array}\right.
\end{align}
\end{thm}

\begin{proof}
We use mathematical induction. Obviously, the theorem holds for $U_{T_1}=\t K_1=U_X$,
i.e., $\partial_X$ is an implectic operator and $\t\gamma_1=U$ is a gradient field.
Now we suppose $\t\gamma_s$ is a gradient field, i.e., $\t\gamma_s'=\t\gamma_s'^*$.
Using the recursive relation \eqref{idKP:recur} we can find
\begin{align*}
\t\gamma_{s+1}=\partial_X^{-1}\t K_{s+1}=\frac{1}{s}\partial_X^{-1}\llb\t K_{s},
\t\sigma_{3}\rlb=\frac{1}{s}\partial_X^{-1}(\partial_X\t\gamma_s'\t\sigma_3-\t\sigma_3'\partial_X\t\gamma_s).
\end{align*}
It then follows from Lemma \ref{lem:Hamilt2} that
\begin{align*}
\t\gamma_{s+1}&=\frac{1}{s}\partial_X^{-1}(\partial_X\t\gamma_s'\t\sigma_3+\partial_X\t\sigma_3'^*\t\gamma_s)
=\frac{1}{s}(\t\gamma_s'\t\sigma_3+\t\sigma_3'^*\t\gamma_s)=\mathrm{grad}\langle\t\gamma_s,\t\sigma_3\rangle,
\end{align*}
where we have also made use of $\t\gamma_s'=\t\gamma_s'^*$ and Lemma \ref{lem:Hamilt1}.
This means if $\t\gamma_s$ is a gradient, so is $\t\gamma_{s+1}$. This also leads to the Hamiltonian \eqref{idKP:H}.

For the non-isospectral case, the proof is similar, but we need to point out that
the recursive relation is broken when $s=3$ because of
\begin{align*}
\llb\t\sigma_3,\t\sigma_3\rlb=0.
\end{align*}
So the mathematical inductive method only works for $s>4$,
while for $s\leq4$ we can directly verify  $\t\omega_s'=\t\omega_s'^*$.
\end{proof}

Recalling the expression \eqref{idKP:res} in Proposition \ref{prop:hierar2}, we immediately have
\begin{cor}\label{cor:grad}
Each  gradient field $\t\gamma_s(U)$ can be expressed through $\c L$ by
\begin{align}\label{idKP:grad}
\t\gamma_s(U)=\underset{P}{\mathrm{Res\,}}\c L^s,~~s=1,2,\cdots.
\end{align}
\end{cor}
This  copes with the results on gradients in \cite{CT,Li}.
Besides,
Theorem \ref{thm:Hamilt} leads to the following corollary.
\begin{cor}\label{cor:Noeth}
The differential operator $\partial_X$ is a Noether operator for both isospectral dKP hierarchy 
\eqref{idKP:hierar2} and non-isospectral dKP hierarchy \eqref{ndKP:hierar2}.
\end{cor}

\begin{proof}
For an arbitrary isospectral equation
\begin{align}\label{eq-s}
U_{T_s}=\t K_s(U),
\end{align}
we can find
\begin{align}\label{idKP:Noeth}
\t K_s'\partial_X+\partial_X\t K_s'^*
=(\partial_X\t\gamma_s)'\partial_X+\partial_X(\partial_X\t\gamma_s)'^*
=\partial_X\t\gamma_s'\partial_X-\partial_X\t\gamma_s'^*\partial_X
=0
\end{align}
due to $\t\gamma_s$ being a gradient field, i.e., $\t\gamma_s'=\t\gamma_s'^*$.
This means $\partial_X$ is a Noether operator of the equation \eqref{eq-s}.
Similarly, by $\t\omega_s'=\t\omega_s'^*$,
\begin{align}\label{ndKP:Noeth}
\t\sigma_s'\partial_X+\partial_X\t\sigma_s'^*=0
\end{align}
also holds, which means $\partial_X$ is also a Noether operator for the non-isospectral dKP hierarchy.
\end{proof}

Finally we give a series of theorems as main results of this section.
\begin{thm}\label{thm:alg3}
The Hamiltonians $\{\c H_l(U)\}$ and $\{\c J_r(U)\}$ described in Theorem \ref{thm:Hamilt} span a Lie algebra with basic structure
\bse\label{dKP:alg3}
\begin{align}
&\lpb\c H_l,\c H_r\rpb=0,\label{dKP:alg3a}\\
&\lpb\c H_l,\c J_r\rpb=l\,\c H_{l+r-2},\label{dKP:alg3b}\\
&\lpb\c J_l,\c J_r\rpb=(l-r)\c J_{l+r-2},\label{dKP:alg3c}
\end{align}
\ese
where the Poisson bracket $\lpb\cdot,\cdot\rpb$ is defined as
\begin{align}\label{def:Poiss2}
\lpb F,G\rpb(U)=\Big\langle\frac{\delta F(U)}{\delta U},\partial_X\frac{\delta G(U)}{\delta U}\Big\rangle
\end{align}
with scalar fields $F$ and $G$ on $\c M$.
\end{thm}

\begin{proof}
Let us  prove \eqref{dKP:alg3b}.
We act $\partial_X^{-1}$ on the both sides of
$\llb\t K_l,\t\sigma_r\rlb=l\,\t K_{l+r-2}.$
As a result, the l.h.s reads
\begin{align*}
\partial_X^{-1}\llb\t K_l,\t\sigma_r\rlb=\partial_X^{-1}(\partial_X\t\gamma_l'\t\sigma_r-\t\sigma_r'\partial_X\t\gamma_l)=\mathrm{grad}\langle\t\gamma_l,
\t\sigma_r\rangle
\end{align*}
due to \eqref{ndKP:Noeth}, $\t\gamma_l'=\t\gamma_l'^*$ and Lemma \ref{lem:Hamilt1}, and the r.h.s gives
\begin{align*}
\partial_X^{-1}l\,\t K_{l+r-2}=l\,\t\gamma_{l+r-2}.
\end{align*}
Thus on the potential level we have
\begin{align*}
\langle\t\gamma_l,\t\sigma_r\rangle=l\,\c H_{l+r-2}
\end{align*}
This is nothing but the relation \eqref{dKP:alg3b}. \eqref{dKP:alg3a} and \eqref{dKP:alg3c} can be proved in the same way from
$\llb\t K_l,\t K_r\rlb=0$ and $\llb\t\sigma_l,\t\sigma_r\rlb=(l-r)\t\sigma_{l+r-2}$.
We skip details and finish the proof.
\end{proof}

\begin{thm}\label{thm:alg4}
Each equation
\begin{equation}
U_{T_s}=\t K_s(U)
\label{thm:alg41}
\end{equation}
in the isospectral dKP hierarchy \eqref{idKP:hierar1} has two sets of conserved quantities
\begin{align}\label{idKP:cq}
\{\c H_l(U)\},~~\{\c I_r^s(U,T_s)=sT_s\c H_{s+r-2}(U)+\c J_r(U)\}
\end{align}
and they span a Lie algebra with basic structure
\bse\label{dKP:alg4}
\begin{align}
&\lpb\c H_l,\c H_r\rpb=0,\label{dKP:alg4a}\\
&\lpb\c H_l,\c I_r^s\rpb=l\,\c H_{l+r-2},\label{dKP:alg4b}\\
&\lpb\c I_l^s,\c I_r^s\rpb=(l-r)\c I_{l+r-2}^s,\label{dKP:alg4c}
\end{align}
\ese
where $l,r,s\geq 1$ and we set $\c H_0(U)=\c I_0^s(U)=0$. Equation \eqref{dKP:alg4a} means conserved quantities $\{\c H_l(U)\}$ are in involution.
\end{thm}

\begin{proof}
Since $\{\t K_l(U)\}$ and $\{\t\tau_r^s(U,T_s)\}$ are symmetries and $\partial_X$ is a Noether operator of  the equation \eqref{thm:alg41},
both $\{\t\gamma_l(U)=\partial_X^{-1}\t K_l(U)\}$ and $\{\t\vartheta_r(U,T_s)=\partial_X^{-1}\t\tau_r^s(U,T_s)\}$ provide conserved covariants 
for the equation \eqref{thm:alg41}.
Thus, $\{\c H_l(U)\}$ and $\{\c I_r^s(U,T_s)=sT_s\c H_{s+r-2}(U)+\c J_r(U)\}$ are two sets of conserved quantities for \eqref{thm:alg41}
in the light of Proposition \ref{prop:cq1}.
Relations \eqref{dKP:alg4} can be derived from Theorem \ref{thm:alg3} via some combinations, e.g.,
\begin{align*}
\lpb\c H_l,\c I_r^s\rpb=\lpb\c H_l,sT_s\c H_{s+r-2}+\c J_r\rpb=sT_s\lpb\c H_l,\c H_{s+r-2}\rpb+\lpb\c H_l,\c J_r\rpb=l\,\c H_{l+r-2}.
\end{align*}
We finish the proof.
\end{proof}

\section{Conclusions}\label{sec:concl}

In the paper we have investigated integrability properties for the dKP hierarchy, including symmetries, Hamiltonian structures and conserved quantities.
We obtained four Lie algebras which respectively composed by
the flows $\{\t K_l(U)\}$ and $\{\t\sigma_r(U)\}$,
symmetries $\{\t K_l(U)\}$ and $\{\t\tau_r^s(U,T_s)\}$,
Hamiltonians $\{\c H_l(U)\}$ and $\{\c J_r(U)\}$,
and conserved quantities $\{\c H_l(U)\}$ and $\{\c I_r^s(U,T_s)\}$.
The key starting point is to find a Lax triad for constructing master symmetry and non-isospectral dKP hierarchy.
Lax triads also provided simple representations \eqref{dKP:hierar2} for the isospectral and non-isospectral dKP flows.
Besides, the relations of the master symmetry, symmetries, Noether operator and conserved covariants
also played main roles in deriving integrability properties in the paper.
The obtained Lie algebras have the same centerless Kac-Moody-Virasoro structure
as in the normal KP case (compared with the collection results given in \cite{FHTZ}).
This is because the Lax triads in the two cases have same structures.
Finally, we note that in \cite{FHTZ}  exact continuum limits were described for the
KP hierarchy and semi-discrete KP (also known as D$\Delta$KP) hierarchy together with their integrability characteristics.
There is also a semi-discrete dKP hierarchy \cite{Yu},
but so far its integrability properties and the connection (eg. continuum limit) with the continuous dKP hierarchy are not clear.

\section*{Acknowledgments}
This project is  supported by the NSF of China
(No. 11071157), the SRF of the DPHE of China (No. 20113108110002) and
the Project of ``First-class Discipline of Universities in Shanghai''.

\end{document}